\documentclass{amsart}
\usepackage{amsfonts}
\usepackage{graphicx}
\usepackage{hyperref}
\usepackage[table]{xcolor}
\newcolumntype{R}[1]{>{\raggedleft\arraybackslash }b{#1}}
\newcolumntype{L}[1]{>{\raggedright\arraybackslash }b{#1}}
\newcolumntype{C}[1]{>{\centering\arraybackslash }b{#1}}

\newtheorem{theorem}{Theorem}[section]
\theoremstyle{plain}

\newtheorem{corollary}{Corollary}

\newtheorem{proposition}{Proposition}

\numberwithin{equation}{section}

\begin{document}
\Large
\title[Confidence bands]{Strong approximation for the deviation of kernel copula estimators}
\author{$^{(1)}$Diam Ba}
\author{$^{(1,2)}$Cheikh Tidiane Seck}
\author{$^{(1,3,4)}$Gane Samb Lo}


\begin{abstract}
We prove  a uniform in bandwidth law of the iterated logarithm for the maximal deviation of kernel copula estimators from their expectations. We deal especially with the \textit{local linear}, the \textit{mirror-reflection} and the \textit{transformation} estimators. These results are useful for establishing the strong uniform in bandwidth consistency of these kernel estimators.\\

\bigskip \noindent $^{(1)}$ LERSTAD, Gaston Berger University, Saint-Louis,
Senegal.\newline
\noindent $^{(2)}$ Alioune Diop University, Bambey, Senegal.\newline
\noindent $^{(3)}$ AUST - African University of Sciences and Technology, Abuja, Nigeria\newline
$^{(4)}$ LSTA, Pierre and Marie Curie University, Paris VI, France.\newline

\noindent \textit{Corresponding author}. Gane Samb Lo. Email : gane-samb.lo@edu.ugb.sn. Mailing address : 1178 Evanston Dr NW T3P 0J9, Calgary, Alberta, Canada.

\end{abstract}

\maketitle
\section{Introduction}
Let $\left(X,Y\right)$ be a random couple with joint cumulative distribution function $H$ and marginal distribution functions $F$ and $G$. The Sklar's theorem (see \cite{skl}) says that there exists a bivariate distribution function $C$ on $[0,1]^2$  with uniform margins such that
$$H(x,y)=C(F(x),G(y)).$$
The function $C$ is called the copula associated with $(X,Y)$ and couples the joint distribution $H$ with its marginals. If the marginal distribution functions $F$ and $G$ are continuous, then the copula $C$ is unique and we have  for all $(u,v)\in [0,1]^2$,
$$C(u,v)=H(F^{-1}(u),G^{-1}(v)),$$ 
where $F^{-1}(u)=\inf\{x: F(x)\geq u\}$ and $G^{-1}(v)=\inf\{y: G(y)\geq v\}$ are the generalized inverse functions of $F$ and $G$, respectively. \\

There are three main approaches for copula estimation : parametric, semiparametric and nonparametric. The parametric approach assumes parametric models for both the copula and the marginals and then deals with maximum likelihood or moment method estimation Oakes \cite{oak}(1982). Semiparametric estimation specifies a parametric copula while leaving the marginals nonparametric (see, e.g. Genest \textit{et al.}\cite{ggr}(1995)). The nonparametric approch offers the greatest generality and was initiated by Deheuvels \cite{deh} (1979), who proposed an estimator based on a multivariate empirical distribution function and its marginals. Afterward, some kernel smoothed estimators have been proposed in the literature (see for instance \cite{r4},\cite{r8},\cite{r3},\cite{r2},\cite{r7}).\\

In this paper we are interested with the kernel estimators proposed in Chen and Huang \cite{r2}(2007), Gijbels and Mielniczuk \cite{r4}(1990) and Fermanian \textit{et al.} \cite{r3}(2004), res pectively called the local linear, the mirror-reflection and the transformation estimators. We will establish for each of them a uniform in bandwidth law of the iterated logarithm for its deviation. These results allows to study the uniform consistency of kernel copula estimators over compact sets $[a,b]^2$, with $0<a<b<1.$\\

Let $(X_1,Y_1),...,(X_n,Y_n)$ be an independent and identically distributed sample of the bivariate random vector $(X,Y)$ with joint cumulative distribution function $H$ and marginal distribution functions $F$ and $G$. Denote by $F_n$ and $G_n$  the marginal empirical cumulative distribution functions. Let $\phi:[0,1]\mapsto[0,1]$ be an increasing transformation and  $K(\cdot,\cdot)$ a bivariate kernel function. For all $0\leq u,v\leq 1$ and $0<h<1$, define the general estimator
\begin{equation}\label{ge}
\hat{C}_{n,h}^{(\cdot)}(u,v)=\frac{1}{n}\sum_{i=1}^{n}K\left(\frac{\phi^{-1}(u)-\phi^{-1}(\hat{U}_i)}{h},\frac{\phi^{-1}(v)-\phi^{-1}(\hat{V}_i)}{h} \right)
\end{equation}
If $K(x,y)$ is a multiplicative kernel, i.e. $K(x,y)=K(x)K(y)$ and the pseudo-observations are $\hat{U}_i=\frac{n}{n+1}F_n(X_i)$ and $\hat{V}_i=\frac{n}{n+1}G_n(Y_i)$, \eqref{ge} is exactly the transformation estimator which is defined by
\begin{equation}\label{tra}
\hat{C}_{n,h}^{(T)}(u,v)=\frac{1}{n}\sum_{i=1}^{n}K\left(\frac{\phi^{-1}(u)-\phi^{-1}(\hat{U}_i)}{h} \right)K\left(\frac{\phi^{-1}(v)-\phi^{-1}(\hat{V}_i)}{h} \right).
\end{equation}
For the local linear estimator,  suppose first that the marginals $F$ and $G$ are estimated by
\begin{center} 
$\displaystyle \hat{F}_n(x)=\frac{1}{n}\sum_{i=1}^n K\left(\frac{x-X_i}{b_{n1}}\right),\ \ \ \ \ \hat{G}_n(y)=\frac{1}{n}\sum_{i=1}^n K\left(\frac{y-Y_i}{b_{n2}}\right),$
\end{center}  where $b_{n1}$ and $b_{n2}$ are  bandwidths converging to 0, as $n\rightarrow\infty$ and $K$ is the integral of a symmetric bounded kernel function $k$ supported on $[-1,1]$. 
Next, consider the pseudo-observations $\hat{U}_i = \hat{F}_n(X_i)$ and $\hat{V}_i =\hat{G}_n(Y_i)$ and choose a multiplicative $K(x,y)=K_{u,h}(x)K_{v,h}(y)$, for $u,v\in[0,1]$, where $K_{w,h}(x)=\int_{-\infty}^{x}k_{w,h}(t)dt,\; w=u,v$
 with $k_{w,h}$ a local linear version of the kernel $k$ given by
$$
 k_{w,h}(t)=\frac{k(t)\{a_2(w,h)-a_1(w,h)t\}}{a_0(w,h)a_2(w,h)-a_1^2(w,h)}\mathbb{I}\left\{\frac{w-1}{h}<t<\frac{w}{h}\right\},$$
where $ a_j(w,h)=\int_{(w-1)/h}^{w/h} t^j k(t)dt$ for $j=0,1,2$ ; $w\in[0,1]$ and $0<h<1$ is a bandwidth. Finally taking $\phi(s)=t$, the identity function in \eqref{ge}, we obtain the local linear kernel estimator  defined as 
\begin{equation}\label{lli}
\hat{C}_{n,h}^{(LL)}(u,v)=\frac{1}{n}\sum_{i=1}^{n}K_{u,h}\left(\frac{u-\hat{U_i}}{h} \right)K_{v,h}\left(\frac{v-\hat{V_i}}{h} \right).
\end{equation}
  According to Omelka \textit{et al} (2009), the mirror-reflection estimator can be represented as
 \begin{equation}\label{mrf}
\hat{C}_{n}^{(MR)}(u,v)=\sum_{l=1}^9\left[Z_n(l,u,v)-Z_n(l,u,0)-Z_n(l,0,v)+Z_n(l,0,0)\right],
\end{equation}
where
\begin{equation}\label{bigz}
Z_n(l,u,v)=\frac{1}{n}\sum_{i=1}^nK\left(\frac{u-\hat{U}_i^{(l)}}{h_n}\right)K\left(\frac{v-\hat{V}_i^{(l)}}{h_n}\right)
\end{equation}
and 
\begin{multline}
\left\{\left(\hat{U}_i^{(l)},\hat{V}_i^{(l)}\right),i=1,...,n,l=1,...,9\right\}\\
=\left\{\left(\pm\hat{U}_i,\pm\hat{V}_i\right),\left(\pm\hat{U}_i,2-\hat{V}_i\right),\left(2-\hat{U}_i,\pm\hat{V}_i\right),\left(2-\hat{U}_i,2-\hat{V}_i\right),i=1,...,n\right\}\nonumber
\end{multline}
Setting $\phi(t)=t$ and using a multiplicative kernel $K(x,y)=K(x)K(y)$, one can see that each quantity $Z_n(l,u,v)$ may be put in the form \eqref{ge}.\\

The remainder of the paper is organized as follows. In Section 2, we state our main results which consist of laws of the iterated logarithm for the deviations of the estimators \ref{tra},\ref{lli} and \ref{mrf} from their means.  Section 3 is devoted to the proofs of the results. Finally, in Section 4 we give an appendix, where we show a key result for the proofs inspired by a general theorem of Mason and Swanpoel (2010), concerning the uniform in bandwidth consistency of kernel-type function estimators. 

\section{Main results }\label{ssec2}
We state our theoretical results in Theorems \ref{t1}, \ref{t2} and \ref{t3}, which give  uniform in bandwidth laws of the iterated logarithm (LIL) for the maximal deviation of the estimators \eqref{tra},\eqref{lli} and \eqref{mrf} from their expectations.\\
 Let $R_n =\left(\frac{n}{2\log\log n}\right)^{1/2}$. We have
 \begin{theorem}\label{t1}
For any sequence of positive constants $(b_n)_{n\geq 1}$ satisfying $0<b_n<1, b_n\rightarrow 0$, $b_n\geq (\log n)^{-1}$, and for some $c>0$, we have almost surely
\begin{equation}\label{e2}
\limsup_{n\rightarrow\infty}\left\{R_n\sup_{\frac{c\log n}{n}\le h\le b_n}\sup_{(u,v)\in(0,1)^2}\left|\hat{C}_{n,h}^{(LL)}(u,v)-\mathbb{E}\hat{C}_{n,h}^{(LL)}(u,v)\right|\right\}\leq 3.
\end{equation}
\end{theorem} 

\begin{theorem}\label{t2}
Suppose that the increasing transformation $\phi$ admits a bounded derivative $\phi'$. Then, for any sequence of positive constants $(b_n)_{n\geq 1}$ satisfying $0<b_n<1, b_n\rightarrow 0$, $b_n\geq (\log n)^{-1}$, and for some $c>0$, we have almost surely 
\begin{equation}\label{ee2}
\limsup_{n\rightarrow\infty}\left\{R_n\sup_{\frac{c\log n}{n}\le h\le b_n}\sup_{(u,v)\in(0,1)^2}\left|\hat{C}_{n,h}^{(T)}(u,v)-\mathbb{E}\hat{C}_{n,h}^{(T)}(u,v)\right|\right\}\leq 3.
\end{equation}
\end{theorem}
\begin{theorem}\label{t3}
For any sequence of positive constants $(b_n)_{n\geq 1}$ satisfying $0<b_n<1, b_n\rightarrow 0$, $b_n\geq (\log n)^{-1}$, and for some $c>0$, we have almost surely 
\begin{equation}\label{eeq2}
\limsup_{n\rightarrow\infty}\left\{R_n\sup_{\frac{c\log n}{n}\le h\le b_n}\sup_{(u,v)\in(0,1)^2}\left|\hat{C}_{n,h}^{(MR)}(u,v)-\mathbb{E}\hat{C}_{n,h}^{(MR)}(u,v)\right|\right\}\leq 3.
\end{equation}
\end{theorem}

The proofs of Theorems \ref{t1},\ref{t2},\ref{t3} are similar and are postponded until Section 3. They are obtained by combining a general theorem of Mason and Swanpoel \cite{r5}(2010) for proving the uniform in bandwidth consistency of kernel-type function estimators, and a law of the iterated logarithm for Kiefer processes due to Wichura \cite{r10}(1973).\\

\textbf{Remark}.
{\rm \begin{itemize}
\item[1)] Under smoothness conditions on the copula $C$, namely the existence bounded second-order partial derivatives, ensuring the uniform almost sure convergence of the bias of the estimators  to zero, we obtain the strong uniform in bandwidth consistency of the kernel copula estimators $\hat{C}_{n,h}^{(LL)},\hat{C}_{n,h}^{(MR)}$ and $\hat{C}_{n,h}^{(T)}$.
\item[2)] The uniformity on the bandwidth allows the use of a large bandwidth selection methods including the method of shrinkage proposed by Omelka \textit{et al.}\cite{r7}(2009).
\item[3)] As the boundary bias is present, the uniform consistency of these kernel estimators is not valid over the entire $[0,1]^2$, but it is true for all  $(u,v)\in[h,1-h]^2,\; 0<h<1$.
\item[4)] If we consider a data-driven bandwidth ; that is $h=\hat{h}_n$, we have the probability versions of Theorems \ref{t1},\ref{t2} and \ref{t3}, stated in  Corollary \ref{cr1} below.\\
\end{itemize} }

Let $\hat{C}_{n,h}^{(\cdot)}$ denote a kernel estimator of the copula $C$, where $(\cdot)$ stands for $(LL), (MR)$ and $(T)$, respectively. One has
\begin{corollary}\label{cr1}
Under the assumptions of Theorems \ref{t1},\ref{t2} or \ref{t3}, we have for all $0<\epsilon <1$ and for any sequence of data-driven bandwidth $\hat{h}_n$ satisfying \begin{equation} \label{ee4'}
 \mathbb{P}(\frac{c\log n}{n}\le \hat{h}_n\le b_n)\rightarrow 1,\; n\rightarrow\infty,
 \end{equation}
\begin{equation}\label{ee2}
\mathbb{P}\left\{R_n\sup_{(u,v)\in(0,1)^2}\left|\hat{C}_{n,\hat{h}_n}^{(\cdot)}(u,v)-\mathbb{E}\hat{C}_{n,\hat{h}_n}^{(\cdot)}(u,v)\right|>  3(1+\epsilon)\right\}=o(1).
\end{equation}
\end{corollary}
\textbf{Remark}. This result enables us to construct simultaneous confidence bands for the copula curve $C(u,v), 0<u,v<1$, with  asymptotic confidence level $100\%$, under the condition that the bias of the estimator $\mathbb{E}\hat{C}_{n,\hat{h}_n}^{(\cdot)}-C(u,v)$ converges uniformly to zero with the same rate $R_n$, as $n\rightarrow\infty$, i.e.  for all $0<\epsilon <1$,
\begin{equation}\label{e2}
\mathbb{P}\left\{R_n\sup_{(u,v)\in(0,1)^2}\left|\mathbb{E}\hat{C}_{n,\hat{h}_n}^{(\cdot)}(u,v)-C(u,v)\right|>  \epsilon\right\}=o(1).
\end{equation}
An example of such confidence bands is provided in \cite{r1} for the local linear estimator.  
\section{Proofs}
The proofs of the theorems are similar. To simplify we will establish the results for a generalized estimator $\hat{C}_{n,h}^{(\cdot)}$ defined as follows. Let $\phi:[0,1]\mapsto[0,1]$ be an increasing transformation. For a bivariate kernel $K(\cdot,\cdot)$ and $0<u,v<1$, define the estimator
\begin{equation}\label{ge}
\hat{C}_{n,h}^{(\cdot)}(u,v)=\frac{1}{n}\sum_{i=1}^{n}K\left(\frac{\phi^{-1}(u)-\phi^{-1}(\hat{U}_i)}{h},\frac{\phi^{-1}(v)-\phi^{-1}(\hat{V}_i)}{h} \right)
\end{equation}
and set
$$\hat{D}_{n,h}^{(\cdot)}(u,v)=:\hat{C}_{n,h}^{(\cdot)}(u,v)-\mathbb{E}\hat{C}_{n,h}^{(\cdot)}(u,v).$$
We call $\hat{D}_{n,h}^{(\cdot)}(u,v)$ the deviation of the estimator $\hat{C}_{n,h}^{(\cdot)}(u,v)$ from its expectation and we will study its behavior by using general empirical process theory.\\
\indent We remark that if the kernel $K(\cdot,\cdot)$ is multiplicative, i.e. $K(x,y)=K(x)K(y)$, we obtain directly the transformation estimator. If $\phi(t)=t$, the identity function and the local linear kernels $K_{u,h}(\cdot)$ and $K_{v,h}(\cdot)$ are employed, then we obtain the local linear kernel estimator. Note that for the mirror-reflection estimator, we consider the following decomposition due to Omelka \textit{et al} (2009),
 \begin{equation}\label{mirror-local}
\hat{C}_{n}^{(MR)}(u,v)=\sum_{l=1}^9\left[Z_n(l,u,v)-Z_n(l,u,0)-Z_n(l,0,v)+Z_n(l,0,0)\right],
\end{equation}
where
\begin{equation}\label{bigz}
Z_n(l,u,v)=\frac{1}{n}\sum_{i=1}^nK\left(\frac{u-\hat{U}_i^{(l)}}{h_n}\right)K\left(\frac{v-\hat{V}_i^{(l)}}{h_n}\right).
\end{equation}
Setting $\phi(t)=t$ and using a multiplicative kernel, one can see that each quantity $Z_n(l,u,v)$ may be put in the form \eqref{ge}.\\

Let $H_n$, $F_n$ and $G_n$ be the empirical cumulative distribution functions of $H$, $F$ and $G$, respectively. Then  the estimator based directly on Sklar's Theorem is given by
$$ C_n(u,v)= H_n(F_n^{-1}(u),G_n^{-1}(v)),$$
with $F_n^{-1}(u)=\inf\{x:F_n(x)\geq u\}$ and $G_n^{-1}(v)=\inf\{x:F_n(x)\geq v\}$. The bivariate empirical copula process is defined as
$$\mathbb{C}_n(u,v)=\sqrt{n}[C_n(u,v)- C(u,v)],\qquad (u,v)\in [0,1]^2.$$
Introduce the following quantity.
$$ \widetilde{C}_n(u,v)=\frac{1}{n}\sum_{i=1}^n\mathbb{I}\{U_i\leq u,V_i\leq v\}$$
which represents the uniform bivariate empirical distribution function based on a sample $(U_1,V_1),\cdots,(U_n,V_n)$ of i.i.d random  variables uniformly distributed on $[0,1]^2$. 
Define the following empirical process
 $$\mathbb{\widetilde{C}}_n(u,v)=\sqrt{n}[\widetilde{C}_n(u,v)- C(u,v)],\quad (u,v)\in[0,1]^2.$$
Then, wecan easily prove that that 
\begin{equation}\label{c3}
\mathbb{\widetilde{C}}_n(u,v)=\mathbb{C}_n(u,v)+ \frac{1}{\sqrt{n}}.
\end{equation}
For any given increasing transformation $\phi$ and $(u,v)\in [0,1]^2$, define 
\begin{eqnarray*}
&g_{n,h}=\hat{C}_{n,h}^{(\cdot)}(u,v)-\widetilde{C}_n(u,v)&\\
&  = \frac{1}{n}\sum_{i=1}^{n}\left [K\left(\frac{\phi^{-1}(u)-\phi^{-1}(\hat{U}_i)}{h},\frac{\phi^{-1}(v)-\phi^{-1}(\hat{V}_i)}{h} \right)- \mathbb{I}\{U_i\leq u,V_i\leq v\}\right ]&\\
& =\frac{1}{n}\sum_{i=1}^{n}\left [K\left(\frac{\phi^{-1}(u)-\phi^{-1}(\hat{F}_n oF^{-1}(U_i))}{h},\frac{\phi^{-1}(v)-\phi^{-1}(\hat{G}_n oG^{-1}(V_i))}{h} \right)-\mathbb{I}\{U_i\leq u,V_i\leq v\}\right]& \\
& =: \frac{1}{n}\sum_{i=1}^{n}g(U_i,V_i,h),&
\end{eqnarray*}
where $g$ belongs to the class of measurable functions $\mathcal{G}$ defined as 
$$
\mathcal{G}=\left\{\begin{array}{c} 
g:(s,t,h)\mapsto g(s,t,h)=K\left(\frac{\phi^{-1}(u)-\phi^{-1}(\zeta_{1,n}(s))}{h}, \frac{\phi^{-1}(v)-\phi^{-1}(\zeta_{2,n}(t))}{h}\right)- \mathbb{I}\{s\leq u,t\leq v\},\\
 u,v\in[0,1], 0<h<1 \,\text{and }\, \zeta_{1,n}\zeta_{2,n}:[0,1]\mapsto [0,1] \,\text{nondecreasing.}
\end{array} \right\}
$$
Since $\mathbb{E}\widetilde{C}_n(u,v)=C(u,v)$, one can observe that 
$$
\sqrt{n}\vert g_{n,h}- \mathbb{E}g_{n,h}\vert= \vert\sqrt{n}\hat{D}_{n,h}^{(\cdot)}(u,v)-\mathbb{\widetilde{C}}_n(u,v)\vert.
$$

Now, we want to apply the main Theorem in \cite{r5} due to Mason and Swanpoel. Towards this end, the class of functions $\mathcal{G}$ must verify the following  four conditions: 
\begin{itemize}
\item[(G.i)]There exists a finite constant $\kappa>0$ such that
 $$ \sup_{0\leq h\leq 1}\sup_{g\in \mathcal{G}}\left\|g\left(\cdot,\cdot,h\right)\right\|_\infty=\kappa <\infty. $$
\item[(G.ii)]\ \ \ There exists a constant $C'>0$ such that for all $h\in [0,1]$,
$$ \sup_{g\in \mathcal{G}}\mathbb{E}\left[g^2\left(U,V,h\right)\right]\leq C'h. $$
\item[(F.i)]\ \ \ 
$\mathcal{G}$ satisfies the uniform entropy condition, i.e., 
$$\exists \, C_0>0, \nu_0>0\ :\ N\left(\epsilon,\mathcal{G}\right)\leq C_0\epsilon^{-\nu_0}.$$
\item[(F.ii)]\ \ \ $\mathcal{G}$ is a pointwise measurable class, i.e there exists a countable sub-class $\mathcal{G}_0$ of $\mathcal{G}$ such that for all $g\in \mathcal{G}$, there exits $\left(g_m\right)_m\subset \mathcal{G}_0$ such that $g_m\longrightarrow g.$\\
\end{itemize}

The checking of these conditions will be done in Appendix and constitutes the proof of the following proposition.
\begin{proposition}\label{p1}
Suppose that the copula function $C$ has bounded first order partial derivatives on $(0, 1)^2$ and the transformation $\phi$ admits a bounded derivative $\phi'$. Then assuming (G.i), (G.ii), (F.i) and (F.ii), we have for some $c > 0,\ 0 < h_0 < 1,$ with probability one, 
$$
\limsup_{n\rightarrow\infty}\sup_{\frac{c\log n}{n}\leq h \leq h_0}\sup_{(u,v)\in(0,1)^2}
\frac{|\sqrt{n}\hat{D}_{n,h}^{(\cdot)}(u,v)-\tilde{\mathbb{C}}_n(u,v)|}{\sqrt{h(|\log h\vert\vee\log\log n)}}=A(c),
$$
where $A(c)$ is a positive constant.
\end{proposition}

\begin{corollary}\label{crl1}
 Under the assumptions of Proposition \ref{p1}, one has for any sequence of constants $0<b_n<1,$ satisfying $\ b_n\rightarrow 0,\ b_n\geq (\log n)^{-1}$, with probability one, 
 $$
 \sup_{\frac{c\log n}{n}\leq h \leq b_n}\sup_{(u,v)\in(0,1)^2}
\frac{|\sqrt{n}\hat{D}_{n,h}^{(\cdot)}(u,v)-\tilde{\mathbb{C}}_n(u,v)|}{\sqrt{\log\log n}}=O(\sqrt{b_n}).
 $$
 \end{corollary}
\begin{proof}{( \textbf{Corollary \ref{crl1})}}\\
{\rm First, observe that the condition $b_n\geq (\log n)^{-1}$ implies
\begin{equation}\label{dc}
\frac{\vert \log b_n\vert}{\log\log n}\leq 1.
\end{equation}
Next, by the monotonicity of the function $x\mapsto x\vert\log x \vert$ on $[0,1/e]$, one can write for $n$ large enough, $h\vert\log h \vert\leq b_n\vert\log b_n \vert$ and hence,
\begin{equation}
h(\vert\log h \vert  \vee\log\log n)\leq b_n(\vert\log b_n \vert \vee\log\log n).
\end{equation}
Combining this and Proposition \ref{p1}, we obtain 
$$
 \sup_{\frac{c\log n}{n}\leq h \leq b_n}\sup_{(u,v)\in(0,1)^2}
\frac{|\sqrt{n}\hat{D}_{n,h}(u,v)-\tilde{\mathbb{C}}_n(u,v)|}{\sqrt{b_n\log\log n\left(\frac{\vert \log b_n\vert}{\log\log n}\vee 1\right)}}=O(1).
 $$
Thus the Corollary \ref{crl1} follows from  \eqref{dc}.}
\end{proof}
Coming back to the proof of our theorems, we have to show that
\begin{equation}\label{lil}
\limsup_{n\rightarrow\infty}\sup_{\frac{c\log n}{n}\leq h \leq b_n}\sup_{(u,v)\in[0,1]^2}\frac{\left|\sqrt{n}\hat{D}_{n,h}^{(\cdot)}(u,v)\right|}{\sqrt{2\log\log n}}\leq 3.
\end{equation}
Towards this end, we make use of an approximation of the empirical copula process $\mathbb{C}_n$ by a Kiefer process (see e.g., Zari\cite{r11}, page 100). Let $\mathbb{W}(u,v,t)$ be a $3$-parameters  Wiener process defined on $[0,1]^2\times[0,\infty)$. Then the Gaussian process $\mathbb{K}(u,v,t)=\mathbb{W}(u,v,t)-\mathbb{W}(1,1,t).uv$ is called a $3$-parameters  Kiefer process defined on $[0,1]^2\times[0,\infty)$.\\
\indent By Theorem 3.2 in Zari\cite{r11}, for $d=2$, there exists a sequence of Gaussian processes $\left\{\mathbb{K}_{C}(u,v,n), u,v\in[0,1], n>0\right\}$ such that
$$ \sup_{(u,v)\in[0,1]^2}\left|\sqrt{n}\mathbb{C}_n(u,v)-\mathbb{K}_{C}^\ast(u,v,n)\right|=O\left(n^{3/8}(\log n)^{3/2}\right),$$ where 
$$\mathbb{K}_{C}^\ast(u,v,n)=\mathbb{K}_{C}(u,v,n)-\mathbb{K}_{C}(u,1,n)\frac{\partial C(u,v)}{\partial u}-\mathbb{K}_\mathbb{C}(1,v,n)\frac{\partial C(u,v)}{\partial v}.$$
This yields
\begin{equation}\label{c1}
\limsup_{n\rightarrow\infty}\sup_{(u,v)\in[0,1]^2}\frac{\left|\mathbb{C}_n(u,v)\right|}{\sqrt{2\log\log n}}=\limsup_{n\rightarrow\infty}\sup_{(u,v)\in[0,1]^2}\frac{\left|\mathbb{K}_{C}^\ast(u,v,n)\right|}{\sqrt{2n\log\log n}}.
\end{equation}
By the works of Wichura\cite{r10} on the  law of the iterated logarithm , for $d=2$, one has almost surely
\begin{equation}\label{c2}
\limsup_{n\rightarrow\infty}\sup_{(u,v)\in[0,1]^2}\frac{\left|\mathbb{K}_\mathbb{C}^\ast(u,v,n)\right|}{\sqrt{2n\log\log n}}\leq 3,
\end{equation}

\noindent which entails  
$$\limsup_{n\rightarrow\infty}\sup_{(u,v)\in[0,1]^2}\frac{\left|\mathbb{C}_n(u,v)\right|}{\sqrt{2\log\log n}}\leq 3.$$
Since ${\mathbb{C}}_n(u,v)$ and $\tilde{\mathbb{C}}_n(u,v)$ are asymptotically equivalent in view of (\ref{c3}), one obtains
$$\limsup_{n\rightarrow\infty}\sup_{(u,v)\in[0,1]^2}\frac{\left|\tilde{\mathbb{C}}_n(u,v)\right|}{\sqrt{2\log\log n}}\leq 3.$$
Applying Corollary \ref{crl1} and the fact that $\sqrt{b_n}\rightarrow 0$, we obtain \eqref{lil} which proves the theorems.

\section*{Appendix}
\begin{proof}(\textbf{Proposition \ref{p1}})\\
Assume that the function $K(\cdot,\cdot)$ is the integral of a symmetric bounded kernel $k(\cdot,\cdot)$, supported on $[-1,1]^2$, i.e. $K(x,y)=\int_0^x\int_0^y k(s,t)dsdt$. We have to check (G.i), (G.ii), (F.i) and (F.ii).\\

\noindent \textbf{Checking for (G.i):}
Recall that $(U_i,V_i), i\geq 1$ are iid random variables uniformly distributed on $[0,1]^2$, $\zeta_{1,n}(U_i)=\hat{F}_n o F^{-1}(X_i)$ and $\zeta_{2,n}(U_i)=\hat{G}_n o G^{-1}(Y_i)$.
For any function $g\in\mathcal{G}$ and $0<h<1$, we can write 
\begin{eqnarray*}
g\left(U_i,V_i,h\right) & =&K\left(\frac{\phi^{-1}(u)-\phi^{-1}(\zeta_{1,n}(U_i))}{h},\frac{\phi^{-1}(v)-\phi^{-1}(\zeta_{2,n}(V_i))}{h}\right)- \mathbb{I}\{U_i\leq u,V_i\leq v\}\\
& =&\int_{-\infty}^{\frac{\phi^{-1}(u)-\phi^{-1}(\zeta_{1,n}(U_i))}{h}} \int_{-\infty}^{\frac{\phi^{-1}(v)-\phi^{-1}(\zeta_{2,n}(V_i))}{h}}k(s,t)dsdt-\mathbb{I}\{U_i\leq u,V_i\leq v\}\\
& =&\int_{-1}^{1}\int_{-1}^{1}\mathbb{I}\left\{U_i\leq \zeta_{1,n}^{-1}\circ\phi(\phi^{-1}(u)-th) ,V_i\leq \zeta_{2,n}^{-1}\circ\phi(\phi^{-1}(v)-sh)\right\}k(s,t)dsdt \\
 & & -\mathbb{I}\{U_i\leq u,V_i\leq v\}\\
&\leq &\int_{-1}^{1}\int_{-1}^{1}k(t)k(s,t)dsdt-\mathbb{I}\{U_i\leq u,V_i\leq v\}\ \leq 4 \|k\|^2+1,\\
\end{eqnarray*}
where $\displaystyle \|k\|=\sup_{(s,t)\in[-1,1]^2}|k(s,t)|$ represents the supremum norm on $[-1,1]^2$. Thus (G.i) holds by taking $\kappa:=4\|k\|^2+1.$\\

\noindent \textbf{Checking for (G.ii).} We have to show that $\displaystyle \sup_{g\in \mathcal{G}}\mathbb{E}g^2(U,V,h)\leq C_0 h$, where $C_0$ is a positive  constant. One can write
\begin{eqnarray*}\small
&\mathbb{E}g^2(U,V,h) =  \mathbb{E}\left[K\left(\frac{\phi^{-1}(u)-\phi^{-1}(\zeta_{1,n}(U))}{h},\frac{\phi^{-1}(v)-\phi^{-1}(\zeta_{2,n}(V))}{h}\right)- \mathbb{I}\{U\leq u,V\leq v\}\right]^2 &\\
&=  \mathbb{E}\left[K^2\left(\frac{\phi^{-1}(u)-\phi^{-1}(\zeta_{1,n}(U))}{h},\frac{\phi^{-1}(v)-\phi^{-1}(\zeta_{2,n}(V))}{h}\right)\right] &\\
&  - 2\mathbb{E}\left[K\left(\frac{\phi^{-1}(u)-\phi^{-1}(\zeta_{1,n}(U))}{h},\frac{\phi^{-1}(v)-\phi^{-1}(\zeta_{2,n}(V))}{h}\right)\mathbb{I}\{U\leq u,V\leq v\}\right] + C(u,v)&\\
& =:  A -2B + C(u,v)&.
\end{eqnarray*}
 Since the function $K(\cdot,\cdot)$ is a kernel of a distribution function, we may assume without loss of generality that it takes its values in $[0,1]$. Then, we can use the inequality $K^2(x,y)\leq K(x,y)$ to bound up the term $A$ in the right hand side of the previous egality.
\begin{eqnarray*}
&A  =   \mathbb{E}\left[K^2\left(\frac{\phi^{-1}(u)-\phi^{-1}(\zeta_{1,n}(U))}{h},\frac{\phi^{-1}(v)-\phi^{-1}(\zeta_{2,n}(V))}{h}\right)\right]&\\
&\leq  \mathbb{E}\left[K\left(\frac{\phi^{-1}(u)-\phi^{-1}(\zeta_{1,n}(U))}{h},\frac{\phi^{-1}(v)-\phi^{-1}(\zeta_{2,n}(V))}{h}\right)\right]&\\
& \leq \mathbb{E}\left[ \int_{-1}^{1}\int_{-1}^{1}\mathbb{I}\left\{U\leq \zeta_{1,n}^{-1}\circ\phi(\phi^{-1}(u)-sh),V\leq \zeta_{2,n}^{-1}\circ\phi(\phi^{-1}(v)-th)\right\}k(s,t)dsdt \right].&
\end{eqnarray*}
The other term $B$ is written into
\begin{eqnarray*}
&B  = \mathbb{E}\left[K\left(\frac{\phi^{-1}(u)-\phi^{-1}(\zeta_{1,n}(U))}{h},\frac{\phi^{-1}(v)-\phi^{-1}(\zeta_{2,n}(V))}{h}\right)\mathbb{I}\{U\leq u,V\leq v\}\right]&\\
&= \mathbb{E}\left[ \int_{-1}^{1}\int_{-1}^{1}\mathbb{I}\left\{s\leq \frac{\phi^{-1}(u)-\phi^{-1}(\zeta_{1,n}(U))}{h},t\leq \frac{\phi^{-1}(v)-\phi^{-1}(\zeta_{2,n}(V))}{h}\right\} \mathbb{I}\{U\leq u,V\leq v\}k(s,t)dsdt\right]&\\
&= \mathbb{E}\left[ \int_{-1}^{1}\int_{-1}^{1}\mathbb{I}\left\{U\leq u\wedge \zeta_1^{-1}\circ\phi(\phi^{-1}(u)-sh),V\leq v\wedge \zeta_2^{-1}\circ\phi(\phi^{-1}(v)-th)\right\}k(s,t)dsdt \right],&
\end{eqnarray*} where $x\wedge y=\min (x,y)$.
Note that $$ C(u,v)=\int_{-1}^{1}\int_{-1}^{1}C(u,v)k(s,t)dsdt,$$
as the kernel $k(\cdot,\cdot)$ satisfies  $\int_{-1}^{1}\int_{-1}^{1}k(s,t)dsdt=1$. Thus
\begin{eqnarray*}
&\mathbb{E}g^2(U,V,h)\leq \mathbb{E}\left[ \int_{-1}^{1}\int_{-1}^{1}\mathbb{I}\left\{U\leq \zeta_{1,n}^{-1}\circ\phi(\phi^{-1}(u)-sh),V\leq \zeta_{2,n}^{-1}\circ\phi(\phi^{-1}(v)-th)\right\}k(s,t)dsdt \right]&\\
 &  -2\mathbb{E}\bigg[ \int_{-1}^{1}\int_{-1}^{1}\mathbb{I}\left\{U\leq u\wedge \zeta_{1,n}^{-1}\circ\phi(\phi^{-1}(u)-sh),V\leq v\wedge \zeta_{2,n}^{-1}\circ\phi(\phi^{-1}(v)-th)\right\}k(s,t)dsdt \bigg]& \\
 & + \int_{-1}^{1}\int_{-1}^{1}C(u,v)k(s,t)dsdt.&
\end{eqnarray*}
We shall suppose that the empirical kernel distributions $\hat{F}_n$ and $\hat{G}_n$ are asymptotically equivalent the classical empirical distribution function $F_n$ $G_n$. From the Chung's (1949) LIL, we can infer that for all $u\in [0,1]$, as $n\rightarrow\infty$,
$$\zeta_{1,n}^{-1}(u)-u = F\circ\hat{F}_n^{-1}(u)-F\circ F^{-1}(u)=O(n^{-1}\log\log n).$$
That is $\zeta_{1,n}^{-1}(u)$ is asymptotically equivalent to $u$. Same for $\zeta_{2,n}^{-1}(u)=G\circ\hat{G}_n^{-1}(u)$. Thus, for all large $n$, we can write
\begin{eqnarray*}
&\mathbb{E}g^2(U,V,h)\leq \mathbb{E}\left[ \int_{-1}^{1}\int_{-1}^{1}\mathbb{I}\left\{U\leq \phi(\phi^{-1}(u)-sh),V\leq \phi(\phi^{-1}(v)-th)\right\}k(s,t)dsdt \right]&\\
 &  -2\mathbb{E}\left[ \int_{-1}^{1}\int_{-1}^{1}\mathbb{I}\left\{U\leq u\wedge \phi(\phi^{-1}(u)-sh),V\leq v\wedge \phi(\phi^{-1}(v)-th)\right\}k(s,t)dsdt \right]&\\
 &  +\int_{-1}^{1}\int_{-1}^{1}C(u,v)k(s,t)dsdt.&
\end{eqnarray*}
That is,
\begin{eqnarray}\label{eqcas}
\mathbb{E}g^2(U,V,h)&\leq &\int_{-1}^{1}\int_{-1}^{1}C\left(\phi(\phi^{-1}(u)-sh),\phi(\phi^{-1}(v)-th)\right)k(s,t)dsdt \nonumber\\
 & & -2\int_{-1}^{1}\int_{-1}^{1}C\left(u\wedge\phi(\phi^{-1}(u)-sh),v\wedge \phi(\phi^{-1}(v)-th)\right)k(s,t)dsdt \\
 &  & +\int_{-1}^{1}\int_{-1}^{1}C(u,v)k(s,t)dsdt.\nonumber
\end{eqnarray}
Now, we have to discuss condition (G.ii) in the four following cases:\\

\textbf{Case 1.} $u\wedge\phi(\phi^{-1}(u)-sh)=\phi(\phi^{-1}(u)-sh)\ \text{et}\ v\wedge \phi(\phi^{-1}(v)-th)=\phi(\phi^{-1}(v)-th)$.\\
In this case the second member of inequality \eqref{eqcas} is reduced and we have 
\begin{eqnarray*}
\mathbb{E}g^2(U,V,h)&\leq &-\int_{-1}^{1}\int_{-1}^{1}C\left(\phi(\phi^{-1}(u)-sh),\phi(\phi^{-1}(v)-th)\right)k(s,t)dsdt \\
  &  & +\int_{-1}^{1}\int_{-1}^{1}C(u,v)k(s,t)dsdt\\
	&\leq &\int_{-1}^{1}\int_{-1}^{1}\left[C(u,v)-C\left(\phi(\phi^{-1}(u)-sh),\phi(\phi^{-1}(v)-th)\right)\right]k(s,t)dsdt.
\end{eqnarray*}
By a Taylor expansion for the copula function $C$, we have
\begin{eqnarray*}
C(u,v)-C(\phi(\phi^{-1}(u)-sh),\phi(\phi^{-1}(v)-th))&=&[u-\phi(\phi^{-1}(u)-sh)]C_u(u,v)\\
 &+ & [v-\phi(\phi^{-1}(v)-th)]C_v(u,v)+ o(h). 
\end{eqnarray*}
Applying again a Taylor-Young expansion for the function $\phi$, we obtain
$$ u-\phi(\phi^{-1}(u)-sh)=\phi(\phi^{-1}(u)) -\phi(\phi^{-1}(u)-sh)=\phi'(u)sh+o(h)$$
and
$$ v-\phi(\phi^{-1}(v)-sh)=\phi(\phi^{-1}(v)) -\phi(\phi^{-1}(v)-th)=\phi'(u)th+o(h).$$
Thus
\begin{eqnarray*}
\mathbb{E}g^2(U,V,h)&\leq &\int_{-1}^{1}\int_{-1}^{1}\left[\phi'(u)C_u(u,v)h+\phi'(v)C_v(u,v)h\right]k(s,t)dsdt \\
	&\leq & 4h\left[\|C_u\|+\|C_v\|\right]\sup_{x\in[0,1]}|\phi'(x)|\|k\|.
\end{eqnarray*}
Taking $C_0=4\left[\left\|C_u\right\|+\left\|C_v\right\|\right]\left\|\phi'\right\|\left\|k\right\|$ gives condition (G.ii).\\

\textbf{Case 2.} $u\wedge\phi(\phi^{-1}(u)-sh)=u\ \text{et}\ v\wedge \phi(\phi^{-1}(v)-th)=v$. \\
Here the inequality \eqref{eqcas} is reduced to
\begin{eqnarray*}
\mathbb{E}g^2(U,V,h)\leq \int_{-1}^{1}\int_{-1}^{1}\left[C\left(\phi(\phi^{-1}(u)-sh),\phi(\phi^{-1}(v)-th)\right)-C(u,v)\right]k(s,t)dsdt.
\end{eqnarray*}
Using the same arguments as in Case 1, we obtain condition (G.ii):
$\displaystyle \sup_{g\in \mathcal{G}}\mathbb{E}g^2(U,V,h)\leq C_0h$, with $C_0=4\left[\left\|C_u\right\|+\left\|C_v\right\|\right]\left\|\phi'\right\|\left\|k\right\|$.\\

\textbf{Case 3.} $u\wedge\phi(\phi^{-1}(u)-sh)=\phi(\phi^{-1}(u)-sh)\ \text{et}\ v\wedge \phi(\phi^{-1}(v)-th)=v$.\\
Here, inequality \eqref{eqcas} is rewritten into
\begin{eqnarray*}
&\mathbb{E}g^2(U,V,h)\leq \int_{-1}^{1}\int_{-1}^{1}C\left(\phi(\phi^{-1}(u)-sh),\phi(\phi^{-1}(v)-th)\right)k(s,t)dsdt & \\
 & -2\int_{-1}^{1}\int_{-1}^{1}C\left(\phi(\phi^{-1}(u)-sh),v\right)k(s,t)dsdt  +\int_{-1}^{1}\int_{-1}^{1}C(u,v)k(s,t)dsdt &\\
&\leq \int_{-1}^{1}\int_{-1}^{1}\left[C\left(\phi(\phi^{-1}(u)-sh),\phi(\phi^{-1}(v)-th)\right)-C\left(\phi(\phi^{-1}(u)-sh),v\right)\right]k(s,t)dsdt &\\
& -\int_{-1}^{1}\int_{-1}^{1}\left[C\left(\phi(\phi^{-1}(u)-sh),v\right)-C(u,v)\right]k(s,t)dsdt. &
\end{eqnarray*}
By applying successively a Taylor expansion for $C$ and for $\phi$, we get
\begin{eqnarray*}
\mathbb{E}g^2(U,V,h)&\leq &\int_{-1}^{1}\int_{-1}^{1}C_v\left(\phi(\phi^{-1}(u)-sh),\theta_1\right)\left[\phi(\phi^{-1}(v)-th)-\phi(\phi^{-1}(v))\right]k(s,t)dsdt\\
& & -\int_{-1}^{1}\int_{-1}^{1}C_u\left(\theta_2,v\right)\left[\phi(\phi^{-1}(u)-sh)-\phi(\phi^{-1}(u))\right]k(s,t)dsdt\\
&\leq & \int_{-1}^{1}\int_{-1}^{1}C_v\left(\phi(\phi^{-1}(u)-sh),\theta_1\right)\phi'(\gamma_1).(-th)k(s,t)dsdt\\
& & -\int_{-1}^{1}\int_{-1}^{1}C_u\left(\theta_2,v\right)\phi'(\gamma_2).(-sh)k(s,t)dsdt,
\end{eqnarray*}
where $\theta_1\in \left(\phi(\phi^{-1}(v)-th),v\right)\ ;\ \theta_2\in \left(\phi(\phi^{-1}(u)-sh),u\right)\ ;\ \gamma_1\in\left(\phi^{-1}(v)-th,\phi^{-1}(v)\right)\ ;\ \gamma_2\in\left(\phi^{-1}(u)-sh,\phi^{-1}(u)\right).$\\
This implies 
\begin{eqnarray*}
\mathbb{E}g^2(U,V,h)&\leq &4h\left\|C_v\right\|\left\|\phi'\right\|\left\|k\right\|^2\left|t\right|+4h\left\|C_u\right\|\left\|\phi'\right\|\left\|k\right\|\left|s\right|\\
&\leq & 4h\left\|\phi'\right\|\left\|k\right\|\left(\left\|C_v\right\|+\left\|C_u\right\|\right).
\end{eqnarray*}
Thus condition (G.ii) holds, with $C_0=4\left\|\phi'\right\|\left\|k\right\|\left(\left\|C_v\right\|+\left\|C_u\right\|\right).$\\

\textbf{Case 4.} $u\wedge\phi(\phi^{-1}(u)-sh)=u\ \text{et}\ v\wedge \phi(\phi^{-1}(v)-th)=\phi(\phi^{-1}(v)-th)$.\\ 
This case is analogous to Case 3, where the roles of $u$ and $v$ are interchanged. Hence, condition (G.ii) is fulfilled, with the same constant  $C_0=4\left\|\phi'\right\|\left\|k\right\|\left(\left\|C_v\right\|+\left\|C_u\right\|\right).$\\

\textbf{Checking for (F.i)}. We have to check the uniform entropy condition for the class of functions
$$
\mathcal{G}=\left\{\begin{array}{c} 
K\left(\frac{\phi^{-1}(u)-\phi^{-1}(\zeta_{1,n}(s))}{h}, \frac{\phi^{-1}(v)-\phi^{-1}(\zeta_{2,n}(t))}{h}\right)- \mathbb{I}\{s\leq u,t\leq v\},\\
 u,v\in[0,1], 0<h<1 \,\text{and }\, \zeta_{1,n}\zeta_{2,n}:[0,1]\mapsto [0,1] \,\text{nondecreasing.}
\end{array} \right\}
$$
 To this end, we consider the following classes of functions, where $\varphi$ is an increasing function :\\
\noindent
$ \mathbb{F}=\left\{(\varphi(x)+m)/\lambda,\lambda > 0,\;m\in\mathbb{R}\right\}$\\
$\displaystyle \mathbb{K}_0=\left\{K((\varphi(x)+m)/\lambda),\lambda>0, \;m\in\mathbb{R}\right\}$\\
$\displaystyle \mathbb{K}=\left\{K((\phi(x)+m)/\lambda,(\phi(y)+m)/\lambda), \lambda>0,\;m\in\mathbb{R}\right\}$\\
$\displaystyle \mathbb{H}=\left\{K((\phi(x)+m)/\lambda,(\phi(y)+m)/\lambda)-\mathbb{I}\left\{x\leq u,y\leq v\right\}\ ;\lambda> 0, \;m\in\mathbb{R}, (u,v)\in [0,1]^2\right\}$.\\

It is clear that by applying  lemmas 2.6.15 and 2.6.18 in van der Vaart and Wellner (see \cite{r9}, p. 146-147), the sets $\mathbb{F},\;\mathbb{K}_0,\;\mathbb{K},\;\mathbb{H}$ are all VC-subgraph classes. Thus, by choosing  the constant function ${\rm G}(x,y)\mapsto {\rm G}(x,y)=\left\|k\right\|^2+1 $ as an envelope function for the class $\mathbb{H}$ ( indeed ${\rm G}(x,y)\geq \sup_{g\in\mathbb{H}}\left|g(x,y)\right|,\ \forall (x,y))$, we can infer from Theorem 2.6.7 in \cite{r9} that  $\mathbb{H}$ satisfies the uniform entropy condition. Since $\mathbb{H}$ and $\mathcal{G}$ have the same structure, we can conclude that $\mathcal{G}$ satisfies this property too, i.e.
$$\exists \; C_0>0, \nu_0>0\ :\ N\left(\epsilon,\mathcal{G}\right)\leq C_0\epsilon^{-\nu_0},\quad 0<\epsilon<1.$$

\noindent \textbf{Checking for (F.ii).}\\
 Define the class of functions 
$$
\mathcal{G}_0=\left\{\begin{array}{c} 
K\left(\frac{\phi^{-1}(u)-\phi^{-1}(\zeta_{1}(s))}{h}, \frac{\phi^{-1}(v)-\phi^{-1}(\zeta_{2}(t))}{h}\right)- \mathbb{I}\{s\leq u,t\leq v\},\\
 u,v\in[0,1]\cap\mathbb{Q}, 0<h<1 \,\text{and }\, \zeta_{1}\zeta_{2}:[0,1]\mapsto [0,1] \,\text{nondecreasing.}
\end{array} \right\}
$$
It's clear that  $\mathcal{G}_0$ is countable and  $\mathcal{G}_0\subset\mathcal{G}$. Let
$$g(s,t)=K\left(\frac{\phi^{-1}(u)-\phi^{-1}(\zeta_1(s))}{h},\frac{\phi^{-1}(v)-\phi^{-1}(\zeta_2(t))}{h}\right)- \mathbb{I}\{s\leq u,s\leq v\}\in\mathcal{G}, (s,t)\in [0,1]^2$$
and for $m\geq 1$, 
$$g_m(s,t)=K\left(\frac{\phi^{-1}(u_m)-\phi^{-1}(\zeta_1(s))}{h} ,\frac{\phi^{-1}(v_m)-\phi^{-1}(\zeta_2(y))}{h}\right)- \mathbb{I}\{s\leq u_m,t\leq v_m\},$$
where $ u_m=\frac{1}{m^2}[m^2u]+\frac{1}{m^2}$ and $ v_m=\frac{1}{m^2}[m^2v]+\frac{1}{m^2}$.\\
\indent Let $\alpha_m=u_m-u,\quad \beta_m=v_m-v$. Then, we have  $\displaystyle 0<\alpha_m\leq\frac{1}{m^2}$ and $\displaystyle 0<\beta_m\leq\frac{1}{m^2}$. Hence $u_m\searrow u$ and $v_m\searrow v$. By continuity $\phi^{-1}(u_m)\searrow \phi^{-1}(u)$ and $\phi^{-1}(v_m)\searrow \phi^{-1}(v)$. Define
$$\delta_{m,u}=\left(\frac{\phi^{-1}(u_m)-\phi^{-1}(\zeta_1(x))}{h} \right)-\left(\frac{\phi^{-1}(u)-\phi^{-1}(\zeta_1(x))}{h} \right)=\frac{\phi^{-1}(u_m)-\phi^{-1}(u)}{h} $$ and
$$\delta_{m,v}=\left(\frac{\phi^{-1}(v_m)-\phi^{-1}(\zeta_2(y))}{h} \right)-\left(\frac{\phi^{-1}(v)-\phi^{-1}(\zeta_2(y))}{h} \right)=\frac{\phi^{-1}(v_m)-\phi^{-1}(v)}{h}.$$
Then
$\delta_{m,u}\searrow 0$ and $\delta_{m,v}\searrow 0$, 
which are equivalent to $$ \left(\frac{\phi^{-1}(u_m)-\phi^{-1}(\zeta_1(x))}{h} \right)\searrow \left(\frac{\phi^{-1}(u)-\phi^{-1}(\zeta_1(x))}{h} \right)$$
 and $$\left(\frac{\phi^{-1}(v_m)-\phi^{-1}(\zeta_2(y))}{h} \right)\searrow \left(\frac{\phi^{-1}(v)-\phi^{-1}(\zeta_2(y))}{h} \right).$$
By right-continuity of the kernel $K(\cdot,\cdot)$, we obtain
$$\forall (x,y)\in[0,1]^2, g_m(x,y)\longrightarrow g(x,y), m\rightarrow \infty$$
and conclude that $\mathcal{G}$ is pointwise measurable class.\\
\end{proof}

\end{document}